%% file: ms.tex
\documentclass[12pt]{article}
\usepackage{amssymb}
\usepackage{amsmath}
\usepackage{amsfonts}
\usepackage[top=3cm, bottom=3cm, left=2.5cm, right=2.5cm]{geometry}
\usepackage{hyperref}
\usepackage{tikz}
\usepackage{natbib}
\usepackage[utf8]{inputenc}
\usepackage[T1]{fontenc}
\usepackage{xcolor}

\newtheorem{theorem}{Theorem}

\newtheorem{corollary}[theorem]{Corollary}

\newtheorem{lemma}[theorem]{Lemma}
\newtheorem{fact}[theorem]{Fact}

\newtheorem{proposition}[theorem]{Proposition}

\newenvironment{proof}[1][Proof]{\noindent \textbf{#1.} }{\  \rule{0.5em}{0.5em}}
\begin{document}

\title{Envelope theorem and discontinuous optimisation:\\ the case of positioning choice problems}
\author{Jean-Gabriel Lauzier\footnote{We thank Simone Cerreia-Vioglio, Luigi Montrucchio, Massimo Marinacci, Nenad Kos, Massimiliano Amarante, Fabio Maccheroni, Ali Khan, Ruodu Wang, \'Etienne Breton, David Salib, Christoph Wolf, David Ruiz-Gomez, Daniele D'Arienzo and Andrea Pasqualini for their comments. We acknowledge financial support from Bocconi University.}\\
University of Waterloo}
\maketitle

\begin{abstract}
    This article examines differentiability properties of the value function of \textit{positioning choice problems}, a class of optimisation problems in finite-dimensional Euclidean spaces. We show that positioning choice problems' value function is always almost everywhere differentiable even when the objective function is discontinuous. To obtain this result we first show that the Dini superdifferential is always well-defined for the maxima of positioning choice problems. This last property allows to state first-order necessary conditions in terms of Dini supergradients. We then prove our main result, which is an ad-hoc envelope theorem for positioning choice problems. Lastly, after discussing necessity of some key assumptions, we conjecture that similar theorems might hold in other spaces as well.
\end{abstract}

\break

\section*{Introduction}
\input{intro.tex}
\break

\section{Positioning choice problems}
\input{section1.tex}

\section{Statements}
\input{section2.tex}

\section{Discussion}
\input{section3.tex}

\section*{Observations for further research}
\input{conclusion}

\nocite{*}
\bibliographystyle{aer}
\bibliography{bibliography.bib}


\appendix
\section{Omitted definitions}
    \input{AppendixA.tex}
\section{Omitted proofs}
    \input{AppendixB.tex}

\end{document}

%% file: intro.tex
An envelope theorem is a statement about the derivative of value functions. Envelope
theorems have foundational applications in several fields of mathematical analysis, notably
in the calculus of variations and optimal control. As such, they are also fundamental to
the microeconomic analysis of consumer and producer problems (\citet{mas1995microeconomic}). Envelope theorems and their generalizations, as they are currently formulated 
and used, rely on a fundamental hypothesis: the continuity of the objective 
function. However, this assumption is problematic in applications where 
discontinuities are meaningful. Notably, discontinuous functions are essential 
to modeling pivotal economic phenomena such as executive bonuses, tax brackets 
and indivisible capital investments such as power plants.\\

Optimisation problems with discontinuities are not differentiable and thus 
standard first-order necessary conditions cannot be used to find maxima at 
such points. Consequently, prior envelope theorems cannot be used to 
characterize the behaviour of their value function. In this article, to tackles this issue we define a class of optimisation 
problems in finite-dimensional Euclidean spaces for which the value function 
is almost everywhere differentiable even when the objective function is 
discontinuous. These optimisation problems are dubbed \textit{positioning 
choice problems} since they have a straightforward geometrical interpretation 
as a choice of position. To the best of our knowledge this is the first 
article to explicitly examine the value function of optimisation problems with 
discontinuous objective functions.\\

As mentioned, the maxima of discontinuous functions are not differentiable and 
so standard first-order necessary conditions theorems cannot be used. We 
propose a generalisation of first-order conditions using the Dini 
superdifferential. Dini supergradients are well-defined at the maxima of 
positioning choice problems and so is the Dini superdifferential. To the best 
of our knowledge this property is not shared by other notions of 
superdifferentials in the literature. We then prove an "ad-hoc" envelope theorem for positioning choice problems along 
the lines of Danskin's theorem \citep{danskintheory}. A deeper look at the 
theorem shows that it relies extensively on Euclidean spaces being a 
(Dedekind-)complete ordered field, an observation worth investigating.\\ 

We set to do so with examples aimed at breaking down our main result. We first 
present an optimisation problem for which the maxima are isolated points but 
where the value function is still almost everywhere differentiable. This 
highlights that first-order conditions are not essential to the 
characterization of the derivative of value functions. The second example 
shows how differentiability depends on the domain of the value function. 
Piecing those observations together we suggest that envelope theorems might be 
obtained on other spaces which are isomorphic to the reals. We conclude with 
an example on a Riemannian manifold which supporting this conjecture.\\

The presentation is as follows. We first give a brief overview of the 
literature. We further detail important concepts in the beginning of Section 
1. Section 1 also defines positioning choice problems. It contains two 
examples that help build the intuition behind our ad-hoc envelope theorem. 
Variations of these examples will be used while discussing our findings. Theorem 2
 , which provides a generalized notion of first-order 
 necessary conditions for a maximum in the context of positioning choice 
 problems, and Theorem 3
 , the ad-hoc envelope theorem, 
 are presented in Section 2 which serves the core of the article. Section 3 discuss the main results. The presentation is less formal as we use 
examples to highlight certain properties of positioning choice problems. We 
show how continuity and differentiability properties of value functions are 
tightly connected to the completeness of Euclidean spaces. We suggest that 
envelope theorems similar to Theorem 3
might be stated 
for any spaces which are isomorphic to the reals. Avenues for further research are identified in conclusion. Appendix A 
contains the definitions omitted in the text and a relevant statement of 
Rademacher's theorem, while Appendix B contains the omitted proofs.

\subsection*{Related literature}

The essential property of positioning choice problems is that the value
function is always a locally Lipschitz map between two finite-dimensional 
Euclidean spaces, even when the objective function of the optimisation problem 
is discontinuous. By Rademacher's theorem, this implies that the value 
function is almost everywhere differentiable (on open sets, with regard to the 
Lebesgue measure) and therefore relatively "well-behaved".\\

The literature on envelope theorems is well-established. It is largely focused 
on providing increasingly general conditions on objective functions as to 
characterize the derivative of value functions. Recent examples are 
\citet{morand2015nonsmooth} and \citet{Morand2018}. The authors scrutinize 
what they call Lipschitz programs, a large class of parametric optimisation 
problems where the objective function satisfies Lipschitzianity. They show 
that under weak assumptions the Lipschitz property is inherited by the value 
function. Our approach is different but complementary to theirs as we do not 
study parametric problems instead lifting the Lipschitz assumption on the 
objective function while still obtaining Lipschitzianity of the value 
function. Both \citet{morand2015nonsmooth} and \citet{Morand2018} derive many results of 
\citet{milgrom2002envelope} as special cases since the latter assumes the 
absolute continuity of the objective function. Our findings complements the 
results of \citet{milgrom2002envelope} as well as we consider discontinuous 
objective functions. More details on \citet{milgrom2002envelope} are provided 
in the next section.\\

To the best of our knowledge, this article is the first to explicitly examine the 
value function of optimisation problems with discontinuous objective 
functions. Thus, we prioritise clarity and keep notation, assumptions and 
proofs as simple as possible. This comes at the cost of generality. We highlight the relation between our results and classical approaches  in the main text [\citet{danskintheory}, \citet{borwein2010convex}, \citet{clarke2013functional}].\\

We close with a list of articles that relates to the present paper. \citet{Lauzier2021securitydesign} and \citet{Lauzier2021insurancedesign} applies the methodology of positioning choice problems to contracting problems in the context of entrepreneurial financing and insurance. The observation that this paper's tooling could be used to provide statements similar to Topkis's theorem is motivated by the latter's successful application to economic theory as seen in \citet{milgrom1994monotone} and \citet{quah2007comparative}. \citet{sinander2019converse}'s converse envelope theorem does in an opposite direction to this article, re-establishing an intuitive link between the envelope formula and first-order conditions when the objective function is absolutely continuous.

%% file: section1.tex
\subsection*{Preliminaries}
Let $\Theta$ be a set of parameters and let $Y(\theta)$ be a choice set given parameter $\theta$. Let the function $h:Y(\Theta)\times \Theta \rightarrow \mathbb{R}$ be the \textit{objective function}. The problem
\begin{align*}
    \sup_{y(\theta)\in Y(\theta)}h(y(\theta); \theta)
\end{align*}
is an \textit{optimisation problem}. The \textit{optimal choice correspondence} $\sigma: \Theta \rightrightarrows 2^{Y(\theta)}$ is
\begin{align*}
    \sigma(\theta)=\arg\max_{y(\theta)\in Y(\theta)}h(y(\theta); \theta)
\end{align*}
and the \textit{value function} $V:\Theta \rightarrow \mathbb{R}$ is
\begin{align*}
    V(\theta)=h(y^*(\theta);\theta)
\end{align*}
for $y^*(\theta) \in \sigma(\theta)$. An \textit{envelope theorem} is a statement about the rate of change of $V(\theta)$ when $\theta$ changes, often stated in terms of the derivative $\frac{\partial V(\theta)}{\partial \theta}$. Typically, envelope theorems rely on building a \textit{continuous selection}\footnote{A selection $f$ of a correspondence $F$ is a function such that for every $x\in domain(F)$ it is $f(x)\in F(x)$.} $y^*(\theta)$ of $\sigma(\theta)$ to use the chain rule and write 
\begin{align*}
    \frac{\partial V(\theta)}{\partial \theta}= \frac{\partial h(y^*(\theta);\theta)}{\partial \theta}=\frac{\partial h(y^*(\theta);\theta)}{\partial y^*(\theta)}\frac{\partial y^*(\theta)}{\partial \theta}. 
\end{align*} 
This approach relies on the family $\left\{h(\cdot; \theta): \theta \in \Theta\right\}$
being composed of functions satisfying some form of continuity. The simplest envelope theorem such as the one found in \citet{mas1995microeconomic} assumes that this family consists exclusively of differentiable functions. The basic statement of \citet{milgrom2002envelope} assumes that this family consists of absolutely continuous functions\footnote{ Lipschitzianity is assumed in \citet{morand2015nonsmooth}, \citet{Morand2018} and \citet{clarke2013functional}.}. However, continuity can be problematic in application. We show in this articles that it is sometimes possible to make statements about the derivative of a value function even when the objective function is discontinuous.

\subsection*{Definitions}
Consider a slightly different notation to avoid confusion. Let $1\leq n<\infty$ and let $C \subsetneq  \mathbb{R}^n$ be a closed box in $\mathbb{R}^n$ which contains zero\footnote{Precisely, we consider the usual Euclidean space $(\mathbb{R}^n, \vert \vert \cdot \vert \vert)$, where the $\vert \vert \cdot \vert \vert$ denotes the usual Euclidean distance, the metric induced by the $L^2$-norm.}. 
Let the function $f:C\rightarrow \mathbb{R}$ be the "payoff function" of being at position $y\in C$.
Let the function $g:C\times C \rightarrow \mathbb{R}$ be the "cost of movement function" of moving from being at $x\in C$ to $y$. Assume that $g$ is continuous, that $\sup_{y \in C}\vert f(y)\vert < +\infty$ and that $\max_{(x,y) \in C\times C}\vert g(x,y)\vert < +\infty$.  Further assume that the function $g$ is always an induced metric on $\mathbb{R}^n$ (see appendix A). This assumption is not essential to obtain the main results of this paper, but it greatly streamlines the proofs and thus the exposition. Let $x \in C$ be given and let $h(x,y)=f(y)-g(x,y)$. The optimisation problem
\begin{align*}
    \sup_{y\in C} h(x,y)
\end{align*}
is a \textbf{positioning choice problem}.\\

Assume that $f$ is almost everywhere continuous except on a set of points where 
\begin{align*}
    \mathrm{either}\, \limsup_{x\in C,\, x\rightarrow y}f(x)=f(y)\,\, \text{or}\, \liminf_{x\in C,\, x\rightarrow y}f(x)=f(y).
\end{align*}
For $x\in C$ given, $h(x,y)$ is the objective function and the optimal choice correspondence $\sigma: C \rightrightarrows 2^C$ is $$\sigma(x)=\arg\max_{y\in C}h(x,y).$$ A function $y^*:C \rightarrow C$ is a selection of $\sigma$ if for every $x \in C$, $y^*(x) \in \sigma(x)$. Letting $y^*(x)\in \sigma(x)$, the value function $V:C \rightarrow \mathbb{R}$ is
\begin{align*}
    V(x)= h(x,y^*(x))=f(y^*(x)) - g(x,y^*(x)).
\end{align*}

\subsubsection*{Example: the structure of the line at $n=1$}

Let $n=1$, $C=[0,2]$, $f(x)=x\cdot1_{x\geq 1}$ and $g(x,y)=\alpha \vert y-x \vert$ for $1_{x\geq 1}$ the indicator function, $\alpha > 1$ and $\vert \cdot \vert$ the usual Euclidean distance on $\mathbb{R}$. Notice that since $f$ is discontinuous the objective function also is. In other words, the family of functions $\{h(x,\cdot): x\in C\}$ consists exclusively of functions which are discontinuous. Thus, typical envelope theorems cannot be used to characterize the derivative $ \frac{\partial V(x)}{\partial x}$.\\

Notice however that there exists at least one pair $(x,y)$, $y>x$, for which we have
$$f(x)-g(x,x)=f(y)-g(y,x).$$
Since $\alpha >1$ this pair is unique and is $y=1$ and $x=1-1/\alpha$. The optimal choice correspondence is
\begin{align*}
    \sigma(x)= \begin{cases}\{x\} &\quad \text{if }x\in [0,\, 1-1/\alpha) \cup [1, 2] \\ \{1\} &\quad \text{if }x\in (1-1/\alpha, 1) \\ \{x,1\}&\quad \text{if }x=1-1/\alpha \end{cases}
\end{align*}
and so the value function is 
\begin{align*}
    V(x) = \begin{cases}0 &\quad \text{if }x\in [0,\, 1-1/\alpha]\\ 1-\alpha(1-x) &\quad \text{if }x\in (1-1/\alpha,1) \\x &\quad \text{if }x\in [1,2]. \end{cases}
\end{align*}
Clearly $V(x)$ is almost everywhere differentiable on $(0,2)$ except at $x\in \{1-1/\alpha, 1\}$ and has derivative given by
\begin{align*}
    \left. V'(x)\right \vert_{x\in (0,2)\setminus \{1-1/\alpha, 1\}}=\begin{cases}0 &\quad \text{if }x\in (0,\, 1-1/\alpha)\\ \alpha &\quad \text{if }x\in (1-1/\alpha,1) \\1 &\quad \text{if }x\in (1,2) \end{cases}
\end{align*}
Moreover, $V(x)$ is kinked at $x=1-1/\alpha$ with subdifferential given by $\underline{\partial}V(1-1/\alpha)=[0,\alpha]$ and at $x=1$ with superdifferential $\overline{\partial}V(1)=[1,\alpha].$\\

This examples shows that even when the objective function is not particularly "well-behaved" the value function is. Implicitly, this example uses the structure of the line to guarantee that $\sigma$ is upper hemicontinuous which guarantees a simple optimisation problem. Variations of it will be used in Section 3 for further discussions.

\subsubsection*{Example: $n=2$ and an intuitive argument}
The previous example can be generalized to additional dimensions. The next constructive example aims to provide the reader with more intuition about the main results of the article. Let $n=2$, $C=[0,2]^2$, $g(x,y)= \alpha \vert \vert y-x \vert \vert$, $\alpha >1$ and 
\begin{align*}
    f(x_1,x_2)= \begin{cases}1 &\quad \text{if simultaneously }x_1\geq 1 \text{ and }x_2\geq 1\\ 0 &\quad \text{otherwise}. \end{cases}
\end{align*}
Notice that $f$ is almost everywhere continuous except on the set
$$\mathrm{Dis}_f=\{(x,y):x=1, \, y\in [1,2]\}\cup\{(x,y):x=[1,2], \, y=1\}$$
where it is upper semicontinuous. To every point $(x,y)\in Dis_f \setminus \{(1,1)\}$
there exists a unique $(a,b)\neq (x,y)$ such that $$V\left((a,b)\right)= f\left((a,b)\right) - g\left((a,b),(a,b)\right)=0 = f\left((x,y)\right) - g\left((x,y),(a,b)\right).$$
This $(a,b)$ corresponds exactly to the point $x=1-1/\alpha$ in the previous example so that if $(x,y)=(1,y)$ then $(a,b)=(1-1/\alpha,y)$ and vice-versa for $(x,y)=(x,1)$. Let $K$ be the set of points that satisfies the previous condition. The arc of circle with center $(1,1)$ defined by
$$L:=\{(x,y):x\leq 1,y\leq 1, (x,y) \text{ belongs to a circle with center }(1,1)\text{ and radius } 1-1/\alpha\}$$
is also a set of points that satisfies the previous equality. That is, for every $(a,b)\in L$ it is
$$V((a,b))= f((a,b)) - g((a,b),(a,b))=0 = f((1,1)) - g((1,1),(a,b)).$$
Notice that the set $Kink_V:=\mathrm{Dis}_f\cup K\cup L$ contains all the (interior) points of non-differentiability of $V$ and is a set of measure zero with regard to the Lebesgue measure on $\mathbb{R}^2$. Let $int$ denotes the interior of a set. It is clear that $V$ is differentiable on $\mathcal{M}:=int(C)\setminus Kink_V$, with $\nabla V(x)=(0,0)$ for every point $(x,y)$ for which $\sigma((x,y)) = \{(x,y)\}$. Furthermore, whenever $(x,y)\in \mathcal{M}$ and $\sigma((x,y))=\{(a,b)\}$ $(a,b)\neq (x,y)$ then $(a,b)\in \mathrm{Dis}_f$ and the directional derivative of $V$ at $(x,y)$ in the direction of $(a-x,b-y)$ is\footnote{See Section 2 for a definition of the directional derivative.}
$$V'((x,y);(a-x,b-y))=\alpha\vert \vert (a-x, b-y)\vert \vert.$$
Constructing this example bare hands is tedious but insightful. It shows that there are higher dimensional environments where the value functions are "well-behaved" even when the objective functions are not. It is clear that the same problem could be cast in $n=3,4,...,N<+\infty$ dimensions and the value function would still be almost everywhere differentiable.

%% file: section2.tex
Let us start with a lemma that allows for simpler notation. Since $C$ is closed, $f$ is bounded and $g$ is continuous it is easy to show the following:

\begin{lemma}\label{Max_for_Closed}
Under the assumption made in section 1 there always exists an almost everywhere single-valued correspondence $\overline{f}: C \rightrightarrows 2^\mathbb{R}$ and a selection $\overline{f}^*$ of $\overline{f}$ such that 
\begin{align*}
    f=\overline{f}=\overline{f}^*  \quad \text{ almost everywhere}
\end{align*}
and for $\overline{h}(x,y)=\overline{f}^*(y)-g(x,y)$ it holds
\begin{align*}
    \sup_{y\in C} h(x,y) = \sup_{y\in C} \overline{h}(x,y)= \max_{y\in C} \overline{h}(x,y).
\end{align*}
\end{lemma}

The proof follows from setting $$\overline{f}(y)=\left\{z\in \mathbb{R}:\text{ either } \limsup_{x\in C,\, x\rightarrow y}f(x)=z \text{ or } \liminf_{x\in C,\, x\rightarrow y}f(x)=z \text{ or both}\right\}$$
and observing that for every $x\in C$ it is possible to build a selection for which the problem $\sup_{y\in C} \overline{h}(x,y)$ attains its supremum. In light of Lemma 1, the rest of the presentation uses $\max$ instead of $\sup$ and drops the subscript $x\in C$ while taking limits. \\ 

Consider a positioning choice problem $\max_{y\in C} h(x,y)$ satisfying the assumptions made in Section 1. Let $x$ be given and let $y^*\in cor(C)$ be a maximizer of $h(x,\cdot)$, where $cor$ denotes the algebraic interior of a set.\footnote{Recall that $cor(A)=int(A)$ whenever $A\subset \mathbb{R}^n$. See appendix A for more details} Since the objective function $h$ is discontinuous it is possible that the gradient $\nabla h(x,\cdot)$ is ill-defined at $y^*$. Thus, the standard first-order necessary condition 
$$0=\nabla h(x,y^*)$$
is not necessarily meaningful. There is, however, a suitable substitute notion.\\

The \textit{upper Dini derivative}\footnote{The definitions used here follow \citet{clarke2013functional} chapter 11.} of $h(x,\cdot)$ in the direction of $v\in \mathbb{R}^n$ is 
    \begin{align}
        dh(x,y;v) = \limsup_{\substack{w\rightarrow v\\ t \downarrow 0}} \frac{h(x,y+tw)-h(x,y)}{t} \tag{UDD}\label{UDD}
    \end{align}
and the \textit{Dini supperdifferential} of $h(x,\cdot)$ at $y$, denoted by $\partial_d h(x,y)$, is the set of $\zeta \in \mathbb{R}^n$ such that 
\begin{align}
    dh(x,y;v) \leq \langle \zeta, v \rangle \quad \forall v\in \mathbb{R}^n. \tag{DS}\label{DS}
\end{align}
Each $\zeta \in \partial_d$ is a \textit{Dini supergradient} and if $h(x,\cdot)$ is differentiable at $y^*\in \sigma(x)$ then $$\partial_dh(x,y)=\{\nabla h(x,y^*)\}.$$

\begin{theorem}[First-order necessary conditions for an interior solution on $C$] \label{FONC}
Under the assumptions of Section 1, for every given $x\in C$, if $y^*\in cor(C)$ and $ y^* \in \sigma(x)
$ then 
\begin{align}
    0\in \partial_d h(x,y^*) \label{CFOC} \tag{CFOC}
\end{align}
Moreover, if both $f$ and $g$ are differentiable at $y^*$ then
\begin{align}
0=\nabla h(x,y^*) = \nabla(f(y^*)-g(x,y^*))=\nabla f(y^*) - \nabla g(x,y^*).   \label{SFOC} \tag{SFOC}
\end{align}
\end{theorem}

\begin{proof}
We only have to show the case when $h$ is discontinuous at $y^*$ as the other cases follow. Recall that by Lemma \eqref{Max_for_Closed}, we can assume without lost of generality that
\begin{align*}
    \limsup_{y \rightarrow y^*} h(x,y) = h(x,y^*).
\end{align*}
\textbf{Lemma: $\partial_d h(x,y^*)$ is well-defined}\\
\textbf{Proof of the Lemma:} By definition it is $$h(x,y^*)=f(y^*)-g(x,y^*)= \limsup_{y\rightarrow y^*}f(y^*) -g(x,y^*).$$
For every direction $v\in \mathbb{R}^n$ for which there is a discontinuity it holds that 
    $$dh(x,y^*;v) = \limsup_{\substack{w\rightarrow v\\ t \downarrow 0}} \frac{h(x,y^*+tw)-h(x,y^*)}{t} = - \infty$$
    because $h(x,y^*+tw)-h(x,y^*) \rightarrow K<0$, $K$ constant. Since $f$ is almost everywhere continuous there exists a direction $v$ for which $\vert dh(x,y^*;v)\vert $ is finite. As $y^*\in cor(C)$ and $y^* \in \sigma(x)$ it is 
       $$\partial_d h(x,y^*)\neq \emptyset$$
       and $\partial_d h(x,y^*)$ is well-defined. $\Box$\\
       
It follows immediately from $y^*\in \sigma(x)$ that for every $v\in \mathbb{R}^n$ it holds
       $$dh(x,y^*;v) \leq 0 = \langle 0, v \rangle$$
       and $$0\in  \partial_d h(x,y^*),$$ as desired. \end{proof}\\

The proof is almost identical to the proof of necessity of first-order conditions in books on convex analysis and non-smooth optimisation such as \citet{borwein2010convex}. The difference is the use of the Dini superdifferential, which is well-defined at discontinuity points. To simplify the presentation the remainder of this section always consider positioning choice problems satisfying the assumptions of Section 1 unless specified otherwise. The next statement is the ad-hoc envelope theorem. It uses the \textit{Fr\'echet} derivative defined in appendix A. We remind the reader that the \textit{Fr\'echet} derivative is a generalisation of the standard gradient $\nabla$ and that the two notions coincides on finite-dimensional Euclidean spaces.

\begin{theorem}["Ad-Hoc" envelope theorem for positioning choice problems]\label{Metric_on_closed}
Positioning choice problems always have almost everywhere \textit{Fr\'echet} differentiable value functions $V$. 
\end{theorem}

The logic of the proof is instructive. Since $C$ is closed, $f$ is bounded and $g$ is continuous, it suffice to prove the continuity of $V$ to show that it is Lipschitz and obtain, by Rademacher's theorem, that it is almost everywhere \textit{Fr\'echet}. Proving continuity of $V$ is relatively simple but a bit tedious. By the way of contradiction it is assumed that $V(x)$ is discontinuous at a point. Since $g$ is continuous this implies that there exists a converging sequence $x_n \rightarrow x$ along which (a) either $V(x)>V(x_n)$ or $V(x_n)>V(x)$ so (b) the set of maximizers $\sigma(x_n)$ and $\sigma(x)$ are disjoint. For this to happen it must be the case that $x$ is a discontinuity point of $f$. WLOG let $V(x)>V(x_n)$ and consider $y^*(x)\in \sigma(x)$. Since $g$ is continuous we can find a neighborhood around $x$ such that $y^*(x)$ strictly dominates every points of $\sigma(x_n)$, a contradiction.\\

Using Rademacher's theorem also ensure that every points of $V$ which are not \textit{Fr\'echet} differentiable have well-defined directional derivatives\footnote{Recall that the directional derivative of a function $h:C\rightarrow \mathbb{R}$ at $x\in C$ in the direction of $y\in C$ is
\begin{align*}
    h'(x;y):=\lim_{t\downarrow 0} \frac{h(x+ty) -h(x) }{t}
\end{align*}
when the limit exists.} in every direction. This property is summarized in the next corollary.

\begin{corollary}[Properties of the derivative]\label{properties_of_derivative}
The value function of a positioning choice problem always satisfies the following:
\begin{enumerate}
\item Every point of non-differentiability in the sense of \textit{Fr\'echet} have finite directional derivative in every directions;
    \item if $\sigma(x)=\{x\}$ then $V(x)=f(x)$ and $\nabla V(x)=\nabla f(x)$ whenever $\nabla f(x)$ is well-defined; 
    \item if $g(x,y)= \alpha \vert \vert y-x\vert \vert$ for $\alpha \geq 1$ then to every $x\neq y$, $x \in int(C),\,y\in C$, it holds that $\vert V'(x;y-x) \vert \leq \alpha\vert \vert y-x \vert \vert$. Equality holds whenever $y\in \sigma(x)$. 
\end{enumerate}
\end{corollary}

The corollary is an immediate consequence of Rademacher's theorem and its proof does not provide further insights. Intuitively, Theorem \ref{Metric_on_closed} and its corollary provides a limited version of Danskin's theorem (\citet{danskintheory}) in the context of positioning choice problems. The rest of this section gathers results which are useful in application. Following \citet{milgrom2002envelope}, it is sometimes handy to write the value function as an integral.
\begin{fact}[Integral representation of the value function] \label{integral_representation}
Let $n=1$ and assume $C=[0,M]$ for $0<M<+\infty$. There exists a continuous function $v_0:C \rightarrow \mathbb{R}$ such that for every $\Tilde{x}\in C$ it holds
\begin{align*}
    V(\Tilde{x}) = V(0) + \int_0^{\Tilde{x}}  v_0(x)dx. \tag{IRVF} \label{IRVF}
\end{align*}
\end{fact}
This representation is a simple application of the second fundamental theorem of calculus. The statement can be generalized to higher dimensions, but doing so does not provide further insights on the behaviour of positioning choice problems.\\

The assumption that $g$ is an induced metric is useful for presenting properties of positioning choice problems. However, it is restrictive.

\begin{proposition}[Sufficient conditions for monotonicity of the value function]\label{monotonocity}
Let $g$ be $$g(x,y):=\sum_{i=1}^n g_i(y_i-x_i)$$ for $g_i(y_i-x_i)$ continuous, non-decreasing and almost everywhere differentiable except maybe at $0$. Then the value function $V(x)$ is almost everywhere differentiable. Moreover, if $g$ is positive and for every $y_i<0$ it holds $g_i(y_i)= 0$ then $V$ is non-decreasing: if $y>>x$ then $V(y)\geq V(x)$. 
\end{proposition}
Proposition \ref{monotonocity} is an example of a class of functions $g$ which are useful in economics. \citet{Lauzier2021securitydesign} considers the case when $n=1$ and $g(x,y)=\max\{0,y-x\}$ is interpreted as a manipulation technology of the observable profit of a firm.  A simplified version of the narrative goes as follows. A business owner hires a manager to take care of the company. The real profit $x$ randomly realizes at the end of a quarter. The manager observes $x$ and reports $y$, namely, the accounting profit of the quarter. The owner observes only the accounting profit $y$. Before reporting, the manager can manipulate the accounting profit by either burning some money (if $y<x$ then $g(x,y)=0$) or injecting liquidities (if $y>x$ then $g(x,y)=y-x$). Proposition \ref{monotonocity} implies that the manager's pay must be non-decreasing in observable profit $y$ because the value function of the optimisation problem defined by this simple game is monotonic.\\

Many optimisation problems which are important in application are defined on $\mathbb{R}^n$ and not on a closed box $C$. Extending Theorem \ref{Metric_on_closed} to such setting can be challenging. Second-order derivatives are not defined for discontinuous functions, and second-order sufficient conditions for optimality cannot be used (i.e. Karush-Kuhn-Tucker Theorem). Proposition \ref{Positioning choice problems and interior solutions} shows how to extend Theorem \ref{Metric_on_closed} to such setting. The next useful intermediate result is stated independently as a lemma for ease of presentation.

\begin{lemma}[A plane is a plane]\label{Plane is a plane}
 Let $f:\mathbb{R}^n\rightarrow \mathbb{R}$ be a plane and let $g:\mathbb{R}^n\times \mathbb{R}^n \rightarrow \mathbb{R}_+$ be a positive and convex transformation of the Euclidean distance, i.e $g(x,y)=\Tilde{g}(\vert \vert y-x\vert \vert)$ for some positive and strictly convex function $\Tilde{g}:\mathbb{R}_+ \rightarrow \mathbb{R}_+$. Then for every finite $x$ the optimisation problem
\begin{align*}
    \max_{y\in \mathbb{R}^n} h(x,y)
\end{align*}
attains a finite maximum. Moreover, the value function $V(x)$ is a plane parallel to $f(x)$ and hence continuously differentiable with 
$\nabla V = \nabla f$.
\end{lemma}
The proof of Lemma \ref{Plane is a plane} is straightforward and need not be covered.
\begin{proposition}[Positioning choice problems and interior solutions]\label{Positioning choice problems and interior solutions}
Let $f$ and $g$ be as in Lemma \ref{Plane is a plane} with $f$ non-decreasing. Let $\Tilde{f}:\mathbb{R}^n\rightarrow \mathbb{R}$ be a non-decreasing function which is bounded by $f$. Then the value function 
\begin{align*}
    \Tilde{V}(x)=\Tilde{h}(x,\Tilde{y}^*(x))=\Tilde{f}(\Tilde{y}^*(x))- g(x,\Tilde{y}^*(x)) \quad \text{for}\quad \Tilde{y}^*(x)\in \Tilde{\sigma}(x)=\arg\max_{y\in \mathbb{R}^n}\Tilde{h}(x,y)
\end{align*}
is almost everywhere \textit{Fr\'echet} differentiable and is bounded by the function
\begin{align*}
    V(x)=h(x,y^*(x))=f(y^*(x))- g(x,y^*(x)) \quad \text{for}\quad y^*(x)\in \sigma(x)=\arg\max_{y\in \mathbb{R}^n}h(x,y).
\end{align*}
\end{proposition}
The intuition behind proposition \ref{Positioning choice problems and interior solutions} is as follows. For $\Tilde{f}$ a discontinuous function of interest is found a function $f$ which bounds it. The latter is taken continuous and locally bounded so that standard second-order conditions can be used solving problem $\max_y\{f(y)-g(x,y)\}$. The function $g$ must be "sufficiently convex" so that "$g$ crosses both $f$ and $\Tilde{f}$ from below". This is enough to guarantee that
\begin{align*}
    \Tilde{\sigma}(x)\subset \{y\in \mathbb{R}^n: f(y)- g(x,y)\geq 0 \}
\end{align*} and conclude that $\Tilde{V}(x)$ is finite and bounded by $V(x)$. Similarly, the local Lipschitzianity of $\Tilde{V}$ is insured by the local Lipschitzianity of $V$. Rademacher's theorem then guarantees that $\Tilde{V}$ is almost everywhere \textit{Fr\'echet}.

%% file: section3.tex
We start by providing a simple example of a positioning choice problem where Theorem \ref{FONC} does not hold but where the value function can be characterised in a way paralleling Theorem \ref{Metric_on_closed}. The example emphasis the role of $\mathbb{R}^n$ being a \textit{Dedekind-complete ordered field}. It also suggests that statements similar to Theorem \ref{Metric_on_closed} can be obtained for optimisation problems defined on other spaces as well, for example for optimisation problems defined on Riemannian manifolds. We then conclude with an example where this "conjecture" is correct.

\subsection*{Pathological cases}
Positioning choice problems are defined on finite-dimensional Euclidean spaces. This choice is consequential for three reasons. First, finite-dimensionality allows to define maxima in straightforward ways. Second, it allows to use Rademacher's theorem. Third, Theorem \ref{Metric_on_closed} implicitly exploits the properties of Euclidean spaces being \textit{Dedekind-complete ordered fields}. The next example shows why this is important.\\

Let $n=1$, $C=[0,2]$, $g(x,y)= \vert x-y \vert$ and \begin{align*}
    f(y)=\begin{cases}1 &\quad  \text{if } x=1 \\ 0 & \quad \text{otherwise}.\end{cases}
\end{align*} For $x\in [0,2]$ given, the optimisation problem is
$\max_{y\in [0,2]}h(x,y)$, for which it is easy to check that the optimal choice correspondence is $\sigma (x)=\{1\}$ and thus $V(x)=1-\vert1-x\vert$. The value function $V$ is almost everywhere differentiable. The derivative is
\begin{align*}
    V'(x)\vert_{x\in(0,2)\setminus\{1\}}=\begin{cases}1 &\text{if}\quad 0<x<1\\-1 &\text{if}\quad 1<x<2 \end{cases}
\end{align*}

Two observations are immediate. First, the assumption that for every $z\in dom(f)$ it is 
$$\limsup_{y \rightarrow z}f(y)\geq f(z)$$ is necessary to write down Theorem \ref{Max_for_Closed}. Otherwise the Dini superdifferential of $h(x,y^*)$ might not be well-defined for $y^*$ a maximizer. The problem is that no superdifferentials are well-defined for isolated points (the points $h(x,1)$ in the example). However,\textit{ the assumption is not necessary to obtain differentiable value functions}.\\

Second, the derivative $V'(x)$ is directly computed using the chain rule by setting $y^*(x)=1$. In other words, it is possible to characterize the derivative of value functions with the chain rule even for problems where the objective function fails all continuity properties traditionally assumed in the literature.\footnote{For instance, the main envelope theorem of \citet{milgrom2002envelope} relies on the family
\begin{align*}
    \{h(x,y):x\in [0,2] \text{ and }y\in[0,2]\}
\end{align*}
consisting of functions that are absolutely continuous in $y$. The problem naturally fails this assumption.} Those observations demonstrate that properties of value functions such as continuity and differentiability are not necessarily "inherited" from objective functions. Such\textit{ properties are tighten to the space on which the value function} is defined.\\

Let us highlight further the role of Euclidean spaces being Dedekind-complete ordered fields. Consider again the two functions $f$ and $g$ defined above but now assume that $x\in X=\{0,1,2\}$. Theorem \ref{Metric_on_closed} cannot be used to characterize further the value function of this problem, nor does any other envelope theorem we are aware of. However, the value function is easily seen as
\begin{align*}
    V(x)=\begin{cases}1 &\text{ if } x=1\\ 0 &\text{ if } x\in \{0,2\}\end{cases}.
\end{align*}
This function is defined on the discrete space $X$ so its derivative cannot be defined. This difficulty is handled implicitly in positioning choice problems by insuring that the space of "parameters" $X$ is a (subset of a)  Dedekind-complete ordered field. Precisely, the problem stems from the fact that when $X=\{0,1,2\}$ it is $cor(X)=\emptyset$, while the set $cor(X)$ should be non-empty.\footnote{This example also shed lights on the role of finite-dimensionality. In the setting of this paper, the interior $int$ and the algebraic interior $cor$ of a set coincides. This is not true anymore in infinite-dimensional spaces where the two notions can differ dramatically. It might be possible to prove results similar to Theorem \ref{Metric_on_closed} in some infinite-dimensional spaces such as the ones defined in dynamic programming.}

\subsection*{Extension to other spaces: a Riemannian manifold example}
The previous two examples are pathological but informative. Let $X,Y$ be arbitrary sets, let $f:Y\rightarrow \mathbb{R}$ be a function and $g:X\times Y \rightarrow \mathbb{R}$ be an induced metric
on $X\times Y$. Set $h(x,y)$ as before and consider the optimisation problem $ \sup_{y\in Y} h(x,y)$
for $x\in X$ given. What are the minimal assumptions which must be imposed on the set $X$ and $Y$ in order to obtain statements similar to Theorem \ref{Metric_on_closed}?\\

The examples showed that one such minimal assumption is the algebraic interior of the set $X$ being non-empty. Moreover, the examples also showed that the characterization of maxima in Theorem \ref{Max_for_Closed} is not necessary to write down Theorem \ref{Metric_on_closed}. In other words, \textit{Theorem \ref{Metric_on_closed} does not rely in any fundamental ways on topological properties of the Euclidean space.} Thus, we believe that envelope theorems can be obtained for positioning choice problems defined on many other spaces. This belief is supported by the following example of a positioning choice problem defined on a Riemannian manifold.\\

Let $M$ be the ring (circle) obtained by defining the one-point closure of the $(0,2)$ interval. Let again $f(x)=x\cdot1_{x\geq 1}$, $g(x,y)= \vert y -x \vert$ and $h(x,y)= f(y)-g(x,y)$. For $x\in M$ given the optimisation problem is $\max_{y\in M} h(x,y)$. Notice that by the definition of the one-point closure it is $$\limsup_{z\in M, z \rightarrow 2}f(z)=2=\limsup_{z\in M, z \rightarrow 0}f(y).$$
The optimal choice correspondence is $$\sigma(x)=\begin{cases}\{0\} &\text{if}\quad 0\leq x<1 \\ [x,2] &\text{if}\quad 1\leq x<2 \end{cases}$$ 
and so the value function is $V(x)=1 + \vert 1-x\vert$. This function is almost everywhere differentiable with derivative
$$V'(x)\vert_{M\setminus\{0,1\}}=\begin{cases}-1 &\text{if}\quad 0<x<1\\1 &\text{if}\quad 1<x<2. \end{cases}$$

%% file: conclusion.tex
Envelope theorems are currently formulated to rely on the assumption of continuity of the objective function. In cases where discontinuities need to be analysed this assumption, fundamental to current envelope theorems, forces the development of different approaches. Accordingly we defined a class of optimisation problems called positioning choice problems for which the value function is almost everywhere differentiable, even when the objective function is discontinuous. While positioning choice problems makes the study of discontinuous optimisation programs analytically convenient,many assumptions made in the text can and should be relaxed. We conclude with a few observations to help formulate further research.\\

The fact that the Dini superdifferential is well-defined for (interior) maxima is easily seen to be more general than how it is formulated in the text. Further, the proof of Theorem 2 is almost identical to the proof of necessity of first-order conditions found in leading books on nonsmooth optimisation. We are thus positive that the tools developed in Section 2 can readily be applied to more general optimisation problems with discontinuous objective functions. The main pitfall lies in providing a suitable characterization of global maxima, as discontinuous optimisation programs are not concave and Hessian matrices are not well-defined. Proposition 8 suggests a promising two-step approach for unconstrained optimisation problems: first find a suitable "bound" to the discontinuous objective function to guarantees that maximizers are finite, then characterize the solution using appropriate techniques. Many important economic problems are in fact constrained optimisation problems, and important economic results obtain from comparative statics exercises. A classical example is the consumer's problem and the analytical definition of normal, inferior and Giffen goods through the derivative of demand functions. We believe that the methodology developed in this paper could be used to obtain statements close in spirit to Topkis's theorem, but where the supermodularity of the objective function is replaced by a 
suitable monotonicity assumption and where monotonicity properties of maximizers are obtained directly from Dini supergradients. Such statement should in turn allow revisiting classical problems while lifting continuity assumptions, for instance the consumer problem when goods are indivisible.

%% file: AppendixA.tex
The text always considers the measure spaces $(\mathbb{R}^n, \mathcal{B}(\mathbb{R}^n),\lambda_n)$ of finite-dimensions $1\leq n <\infty$, where $ \mathcal{B}(\mathbb{R}^n)$ is the Borel $\sigma$-algebra of $\mathbb{R}^n$ and $\lambda_n$ is the $n$-dimensional Lebesgue measure. This appendix always considers functions $f:\mathbb{R}^n\rightarrow\mathbb{R}$.\\

A property $P$ of $f$ is said to \textbf{hold almost everywhere} if there exists a set $A \in \mathcal{B}(\mathbb{R}^n)$ such that $\lambda_n(A)=0$ and $P$ is true on $\mathbb{R}^n\setminus A$. \\

A positive function $p: \mathbb{R}^n \rightarrow \mathbb{R}_+$ is said to be a \textbf{norm} if for every $\alpha \in \mathbb{R}$ and every $x,y\in \mathbb{R}^n$ it holds
   \begin{enumerate}
    \item $p(x)=0$ if and only if $x=\mathbf{0}$;
    \item $p(x+y)\leq p(x)+p(y)$;
    \item $p(\alpha x)=\vert \alpha \vert p(x)$.
    \end{enumerate}
A function $g: \mathbb{R}^n\times \mathbb{R}^n \rightarrow \mathbb{R}$ is \textbf{a metric on $\mathbb{R}^n$} if for every $x,y, z \in \mathbb{R}^n$ it satisfies the following conditions:
    \begin{itemize}
    \item \textit{non-negativity}: $g(x,y)\geq 0 $;
    \item \textit{identity of indiscernibles}: $g(x,y)=0 \iff x=y$;
    \item\textit{symmetry}: $g(x,y)=g(y,x)$ and
    \item \textit{subadditivity} or the \textit{triangle inequality}: $g(x,z)\leq g(x,y) +g(y,z)$.
    \end{itemize}
A metric $g$ is said to be \textbf{induced} by the norm $p$ if for every $x, y \in \mathbb{R}^n$ it is $$g(x,y)= p(y-x).$$
Throughout the text, the notation $\vert\vert \cdot \vert \vert$ denotes the usual Euclidean metric induced by the $L^2$-norm $$\left \vert\vert x \vert \right \vert_2=\sqrt{\sum_{i=1}^n x_i^2}.$$

The function $f$ is said to be \textbf{Lipschitz} if there exists a $0\leq K<\infty$ such that for every $x,y \in \mathbb{R}^n$ it holds
\begin{align*}
    \vert f(x)-f(y)\vert \leq K\vert\vert x-y \vert \vert.
\end{align*}

Let $V$ be a normed space. A point $\theta_0 \in \Theta \subset V$ is in the \textbf{algebraic interior} of $\Theta$, denoted by $cor(\Theta)$, if for every $v \in V$ there exists $\eta_\theta > 0$ such that $\theta_0+tv \in \Theta$ for all $0\leq t <\eta_\theta$. If $V= \mathbb{R}^n$ then $cor(\Theta)= int(\Theta)$, where $int$ denotes the interior of a set.\\

Let $V$ and $W$ be normed vector spaces with $\vert\vert \cdot \vert\vert_V$ and $\vert\vert \cdot \vert\vert_W$ the induced metric of $V$ and $W$ respectively and let $O\subset V$ be open. The function $f:O \rightarrow W$ is \textbf{\textit{Fr\'echet} differentiable} at $x\in O$ if there exists a bounded linear operator $A: V \rightarrow W$ such that 
\begin{align*}
    \lim_{\vert\vert h \vert\vert_V \rightarrow 0} \frac{\vert\vert f(x+h) - f(x) -Ah \vert\vert_W}{\vert\vert h \vert\vert_V} =0.
\end{align*}
If $V=\mathbb{R}^n$ and $W=\mathbb{R}$ then \textit{Fr\'echet} differentiability is the "usual" differentiability and $A$ is the gradient $\nabla f(x)$ of $f$ at $x$.\\

A proof of Rademacher's theorem requiring minimal knowledge of measure theory can be found in \citet{borwein2010convex} (Theorem 9.1.2). It is stated here for completeness.

\begin{theorem}[Rademacher's theorem] \label{Rademacher} Let $O\subset \mathbb{R}^n$ be an open subset of $\mathbb{R}^n$ and let $f:O \rightarrow \mathbb{R}$ be Lipschitz. Then the function $f$ is almost everywhere Fr\'echet differentiable.\\
\end{theorem}

%% file: AppendixB.tex
\noindent\textbf{Proof of Lemma \ref{Max_for_Closed}:} Recall that $f$ is almost everywhere continuous except on a set of points where it satisfies  
\begin{align*}
   \mathrm{either }\, \limsup_{x\in C,\, x\rightarrow y}f(x)=f(y) \quad \mathrm{or} \quad   \liminf_{x\in C,\, x\rightarrow y}f(x)=f(y)
\end{align*}
and set the correspondence $\overline{f}: C \rightrightarrows 2^\mathbb{R}$ as 
$$\overline{f}(y)=\left\{z\in \mathbb{R}:\text{ either } \limsup_{x\in C,\, x\rightarrow y}f(x)=z \text{ or } \liminf_{x\in C,\, x\rightarrow y}f(x)=z \text{ or both}\right\}.$$
The set \begin{align*}
    D:=\{y\in C:\limsup_{x\in C,\, x\rightarrow y}f(x) \neq  \liminf_{x\in C,\, x\rightarrow y}f(x)\}
\end{align*}
is a set of measure zero and $\overline{f}=f$ on $C\setminus D$ so the two are almost everywhere equal.\\

Let $x$ be given. The problem $\sup_{y\in C} h(x,y)$ attains a finite supremum because $C$ is a closed box and it is assumed that
\begin{align*}
    \sup_{y \in C}\vert f(y)\vert < +\infty \quad \text{and} \quad \max_{(x,y) \in C\times C}\vert g(x,y)\vert < +\infty.
\end{align*} If this supremum is attained for some $y\in C$ then we are done.\\

Suppose it is not. Then $\sigma(x):=\arg\max_{y\in C}h(x,y) \subset D$. Pick a $y^* \in \sigma(x)$ and observe that since $g$ is continuous it holds
\begin{align*}
   \sup_{y\in C} h(x,y)= \limsup_{y\in C,y\rightarrow y^*} \left (f(y)-g(x,y)\right) = \limsup_{y\in C,y\rightarrow y^*}f(y) - \limsup_{y\in C,y\rightarrow y^*}g(x,y)= \limsup_{y\in C,y\rightarrow y^*}f(y)-g(x,y^*).  
\end{align*}
Let 
\begin{align*}
    \overline{f}^*(y)=\begin{cases}f(y)& \text{ if }y \notin D\\  \max\overline{f}(y)& \text{ if } y \in D.\end{cases}
\end{align*}
The function $\overline{f}^*$ is a selection of $\overline{f}$ for which $f=\overline{f}=\overline{f}^*$ almost everywhere and for $\overline{h}(x,y)=\overline{f}^*(y)-g(x,y)$ it holds
\begin{align*}
   \sup_{y\in C} \overline{h}(x,y)= \max_{y\in C} \overline{h}(x,y) =  \sup_{y\in C}, h(x,y), 
\end{align*} 
as desired. $\blacksquare$\\

\noindent\textbf{Proof of Theorem \ref{Metric_on_closed}:} It suffice to show that $V$ is continuous, as it follows immediately by observing that since $C$ is a closed box, $g$ is a metric and $f$ is bounded the value function is also Lipschitz and hence, by Rademacher's theorem, almost everywhere differentiable.\\

Suppose, by the way of contradiction, that $V$ is discontinuous at $x\in C$. Then there exists an $\alpha>0$ such that for every sequences $(x_n)_{n\in\mathbb{N}} \subset C$ converging to $x$, $x_n\rightarrow x$, it is $\vert V(x)- V(x_n) \vert \geq \alpha$.
By definition this means that
    $$\left \vert f(y^*(x))-g(x,y^*(x))-f(y^*(x_n))+g(x_n,y^*(x_n)) \right \vert \geq \alpha $$
    for every selection $y^*(x)\in \sigma(x)$ and every selection $y^*(x_n)\in \sigma(x_n)$. Since $g$ is continuous, the previous implies that for $n$ large it is $\sigma(x)\cap \sigma(x_n)=\emptyset$.
Without loss of generality suppose that $V(x)>V(x_n)$ and pick a $y^*(x)\in \sigma(x)$. Since $g$ is continuous it holds that
$$\vert V(x)-f(y^*(x))-g(x_n,y^*(x))\vert \rightarrow 0.$$
Thus, there exists a $\beta>0$ such that for every $x_n$ satisfying $\vert\vert x-x_n \vert \vert<\beta$ it is
$$f(y^*(x))-g(x_n,y^*(x))>f(y^*(x_n))+g(x_n,y^*(x_n)),$$
for every $y^*(x_n)\in \sigma(x_n)$, a contradiction. $\blacksquare$\\

\noindent\textbf{Proof of Fact \ref{integral_representation}:} Equality \eqref{IRVF} follows immediately from the second fundamental theorem of calculus by observing that since $V$ is Lipschitz it is almost everywhere differentiable (by Rademacher's theorem) and there exists a Riemann integrable function $v_0$ such that for every $a,b\in C$, $a<b$, it holds
\begin{align*}
    V(b)-V(a)= \int_a^b v_0(x)dx.
\end{align*}
Continuity of $v_0$ is also ensured by the Lipschitzianity of $V$. $\blacksquare$\\

\noindent\textbf{Proof of Proposition \ref{monotonocity}:} As $C$ is closed, $f$ is bounded and $g_i$ are continuous it holds that $V$ is Lipschitz and almost everywhere differentiability follows from Rademacher's theorem. The last statement is trivial if $f$ is non-decreasing. If $f$ is decreasing somewhere and if $g_i(y)= 0$ for every $y<0$ then the last statement follows by observing that there always exists an open ball $B(x)$ such that there exists a $y<<x$, $y\in B(x)$, for which
$$f(y) = f(y)-0 = f(y)- g(y,x) >f(x)-g(y,x) = f(x)$$
and $V$ is non-decreasing. $\blacksquare$\\

\noindent\textbf{Proof of Lemma \ref{Plane is a plane}:} Let $\Tilde{x}\in \mathbb{R}^n$ be given. Under the assumption that $g$ is  a positive and convex transformation of the Euclidean distance  $\vert \vert \cdot \vert \vert$  it holds that the optimisation problem $\sup_{y\in \mathbb{R}^n} h(\Tilde{x},y)$ attains a finite supremum. Let $y^*\in \sigma(\Tilde{x})$ and define $z=y^* - \Tilde{x}$. Since the Euclidean distance is translation invariant $g$ also is. Since $f$ is a plane it holds that for every $x\in \mathbb{R}^n$ we have $x+z\in \sigma (x)$. Hence, $V$ is a plane parallel to $f$. $\blacksquare$\\

\noindent\textbf{Proof of Proposition \ref{Positioning choice problems and interior solutions}:} Let $x$ be given and consider the sets
\begin{align*}
    Z(x):=\{y\in\mathbb{R}^n: f(x,y)-g(x,y) \geq 0\} \text{ and }  \Tilde{Z}(x):=\{y\in\mathbb{R}^n: \Tilde{f}(x,y)-g(x,y) \geq 0\}.
\end{align*}
 Since $f$ dominates $\Tilde{f}$ it holds that $\Tilde{Z}(x)\subset Z(x)$. Thus, the problem $\max_{y\in \mathbb{R}^n}\Tilde{h}(x,y)$ admits interior solutions and by definition it holds that $\Tilde{\sigma}(x)\subset\Tilde{Z}(x)\subset Z(x).$ This implies that the value function $V$ dominates the value function $\Tilde{V}$ and since the former is a plane, the latter is a locally Lipschitz map between two Euclidean spaces. By Rademacher's theorem, $V$ is almost everywhere \textit{Fr\'echet} differentiable, as desired. $\blacksquare$